\documentclass[%
 reprint,
superscriptaddress,
 amsmath,amssymb,
 aps,
 pra,
]{revtex4-2}

\usepackage[english]{babel}

\usepackage[utf8]{inputenc}

\usepackage{graphicx}

\usepackage{multirow}

\usepackage{xcolor}

\usepackage{bm}

\usepackage{amsthm}

\newcommand{\tr}{\text{Tr}}

\newtheorem{theorem}{Theorem}
\newtheorem*{remark}{Remark}

\newtheorem{thm}{Theorem}[section]
\newtheorem{lemma}[thm]{Lemma}
\newtheorem{corollary}[thm]{Corollary}

\begin{document}

\title{Flatness-Preserving Operations}
\author{O. Veltheim}
\email{otto.veltheim@helsinki.fi}
\affiliation{Department of Physics,  P.O.Box 64, FIN-00014 University of Helsinki, Finland}
\author{E. Keski-Vakkuri}
\affiliation{Department of Physics,  P.O.Box 64, FIN-00014 University of Helsinki, Finland}
\affiliation{InstituteQ, the Finnish Quantum Institute, Helsinki, Finland}
\affiliation{Helsinki Institute of Physics, P.O.Box 64, FIN-00014 University of Helsinki, Finland}

\begin{abstract}
A quantum state is called flat if it is proportional to a projector. There has been recent
interest in studying antiflatness, the property of diverging from flat states, and establishing
a resource theory for it. Identifying the free operations (the Flatness-Preserving Operations (FPOs)) remained an open problem. 
For this purpose, we first discuss Orthogonality-Preserving Operations (OPOs), of which trivial examples are unitary operations in an isolated system. More generally, we give a simple proof that all OPOs are isometric embeddings consisting of combinations of unitaries/isometries and appending a fixed state. 
 We then show that all FPOs are either constant maps to some fixed flat state or a special case of an OPO, where the appended fixed state must be a flat state. We also show that the only possible flat convex combinations of flat states are those of orthogonal flat states with weights given by their purities. 
\end{abstract}

\maketitle

\section{Introduction}

\textit{Antiflatness} of quantum states is a relatively novel concept in quantum information theory. A state is called \textit{flat} if it has a flat spectrum, meaning that all non-zero eigenvalues are equal. In a bipartite system $AB$, a state is called \textit{flat$_A$} if its reduced state $\rho_A$ is flat. In this way, flatness can be seen as a property of the entanglement spectrum. The complement of flat states are the antiflat states. Since these concepts can be used to refer to properties of the entanglement spectrum, they characterize features of the entanglement present in the global state. On the other hand, antiflatness also connects features of the entanglement spectrum to non-stabilizerness: entangled stabilizer states of a bipartite qubit system AB are a subset of flat$_A$ states \cite{PhysRevA.109.L040401}. Hence, antiflat$_A$ entanglement spectrum is a signature of non-stabilizer states. It was also found to be a characteristic feature of physical states in quantum field theories \cite{Benedetti:2026mfy}.

Flat and antiflat states also play a role in quantum gravity. In holography, so-called fixed-area states \footnote{Semiclassical states projected to a subspace constrained to imply a fixed area for the Ryu-Takayanagi or Hubeny-Rangamani-Takanayagi surface.} have a flat spectrum to leading order in $G_N$ \cite{Dong:2018seb}. This connects antiflat spectrum to quantum gravitational fluctuations and backreaction. Another observation is that antiflatness lower-bounds the amount of non-local magic, which was also found to be connected to gravitational backreaction in \cite{Cao:2024nrx}.

Since antiflatness is emerging as an interesting quantum informational resource, a resource theory approach for it was developed in \cite{jasser2026journeyflatlanddoesantiflatness}. (A more complete review and list of references of antiflatness and related aspects can be found therein.) There, flat states were identified as resource-free states, flatness-preserving operations (FPOs) were defined as free operations, various antiflatness monotones were reviewed, antiflatness-majorization was defined as a preorder, connected to monotone inequalities and FPOs in an analogous manner as Nielsen's theorem \cite{PhysRevLett.83.436}
connects the more familiar majorization to LOCCs. However, examples of FPOs and in particular their classification remained an open problem. 

This work provides a missing piece by identifying and classifying FPOs for finite-dimensional bipartite systems. We first review what are orthogonality-preserving operations (needed for proofs), then list five theorems that are our main results, identifying the FPOs. Proofs are found in Appendices.

\section{Orthogonality-preserving operations}

We say that two quantum states $\rho,\sigma\in\mathcal{H}_1$ are orthogonal if they are orthogonal with respect to the Hilbert-Schmidt inner product:
\begin{equation}
\tr(\rho\sigma)=0.
\end{equation}
Another way of defining orthogonality would be that the two states should have orthogonal support, i.e., $\rho\sigma=0$. These two definitions are equivalent in the case of quantum states, but not for more general linear operators. Since we are only interested in the orthogonality of quantum states, either definition works in our case .

Any quantum operations or quantum channels discussed in this work are assumed to be completely positive and trace-preserving (CPTP). We will also assume that the dimensions of any Hilbert spaces are finite.

We will call a CPTP quantum channel $\mathcal{E:D(H}_1)\rightarrow\mathcal{D(H}_2)$ an Orthogonality-Preserving Operation (OPO), if
\begin{equation}
\rho\sigma=0\Rightarrow \mathcal{E}(\rho)\mathcal{E}(\sigma)=0\quad \forall\rho,\sigma\in\mathcal{D(H}_1).
\end{equation}
Intuitively, one might already suspect that the only quantum channels which preserve orthogonality are unitary maps. We show that this is indeed the case when the codomain of the quantum channel is the same as the domain. For the more general case, where the codomain can be different, we show that the image of the quantum channel must be isomorphic to appending the domain with a fixed quantum state. It follows that the dimension of the codomain is lower-bounded by the dimension of the domain for OPOs.

The results we prove here for OPOs could possibly also be derived from more general results for $C^*$-algebras \cite{winter2009completelypositivemapsorder}. However, we will provide here simpler proofs for quantum channels, in order to use these results also for the Flatness-Preserving Operations in the next section.

We start by showing that for an OPO, the state of the environment after the channel does not depend on the initial state of the system.

\begin{theorem}
\label{th:OPO}
Let $\mathcal{E:D(H}_1)\rightarrow\mathcal{D(H}_2)$ be an OPO. Let $\mathcal{H}_e$ be an ancillary Hilbert space and $W:\mathcal{H}_1\rightarrow\mathcal{H}_2\otimes\mathcal{H}_e$ an isometry such that
\begin{equation}
\mathcal{E}(\rho)=\tr_{\mathcal{H}_e}(W\rho W^\dagger).
\end{equation}
Then, there exists a state $\sigma\in \mathcal{D(H}_e)$ such that
\begin{equation}
\tr_{\mathcal{H}_2}(W\rho W^\dagger)=\sigma\quad\forall\rho\in\mathcal{D(H}_1).
\end{equation}
\end{theorem}

We then show that OPO implies that the channel is a combination of appending a fixed state and an isometry. The proofs for both of these theorems are given in Appendix \ref{Appendix: OPO}.

\begin{theorem}
\label{th:OPO2}
Let $\mathcal{E:D(H}_1)\rightarrow\mathcal{D(H}_2)$ be an OPO. Then there exists an ancillary Hilbert space $\mathcal{H}_a$ with a state $\sigma\in\mathcal{D(H}_a)$ and an (injective but not necessarily bijective) isometry $V:\mathcal{H}_1\otimes\mathcal{H}_a\rightarrow \mathcal{H}_2$ such that
\begin{equation}
\label{eq:isometry}
\mathcal{E}(\rho)=V\rho\otimes\sigma V^\dagger.
\end{equation}
\end{theorem}

\begin{remark}
 Although the matrices $\sigma$ in Theorem \ref{th:OPO} and $\sigma$ in Theorem \ref{th:OPO2} serve different purposes in their respective theorems, and act on different Hilbert spaces, the two matrices are actually the same, at least in the sense that the non-zero parts of their spectra are identical.
\end{remark}

Note that $\mathcal{H}_a$ can also be 1-dimensional, in which case the isometry would be just $V:\mathcal{H}_1\rightarrow \mathcal{H}_2$. In particular, when $\mathcal{H}_1=\mathcal{H}_2$, the above result means that the only OPOs are the set of unitary operations. In the more general case, the most natural OPOs are perhaps those which append the original state with another fixed state and then perform a unitary operation in the compound Hilbert space:
\begin{equation}
\mathcal{E}(\rho)=U\rho\otimes\sigma U^\dagger.
\end{equation}

\section{Flatness-preserving operations}

We will say that a state $\rho\in\mathcal{D(H)}$ is flat if it is proportional to a projector, i.e., all of its non-zero eigenvalues are the same. Similarly to \cite{jasser2026journeyflatlanddoesantiflatness}, we will say that a state $\omega\in\mathcal{D(H}_A\otimes\mathcal{H}_B)$ in a bipartite system is flat$_A$, if the reduced density operator $\tr_B(\omega)=\omega_A$ is flat. As shown by \cite{jasser2026journeyflatlanddoesantiflatness} and as we discuss in Appendix \ref{Appendix: Flat properties}, some equivalent ways of defining a flat state $\rho$ are also
\begin{equation}
\tr(\rho^3)-\tr(\rho^2)^2=0
\end{equation}
or
\begin{equation}
\rho^2=\tr(\rho^2)\rho.
\end{equation}
It is also good to note that for a flat state $\rho$, the purity $\tr(\rho^2)$ is equal to the only non-zero eigenvalue.

Flat states do not form a convex set, so it is interesting to study, under which conditions is a convex combination of two flat states also a flat state. It turns out that the convex combination is a flat state if and only if the two original flat states are orthogonal and the weights of the two are chosen based on their respective purities.

\begin{theorem}
\label{th:convex}
Let $\rho_1, \rho_2\in \mathcal{D(H)}$ be two flat states and $\rho_1\neq\rho_2$ and let $0<p<1$. Then, the state
\begin{equation}
\sigma=p\rho_1+(1-p)\rho_2
\end{equation}
is flat if and only if
\begin{equation}
\rho_1\rho_2=0,
\end{equation}
(i.e., $\rho_1$ and $\rho_2$ have orthogonal support) and
\begin{equation}
p\tr(\rho_1^2)-(1-p)\tr(\rho_2^2)=0.
\end{equation}
\end{theorem}

The proof is given in Appendix \ref{Appendix: Convex combinations}. Intuitively, this result makes sense, of course, since, for example, for a two-dimensional Hilbert space, the only flat states are clearly the pure states and the maximally mixed state.

In this work, we will use two alternative definitions of a Flatness-Preserving Operation (FPO). Our main definition is defined analogously to the OPOs: A CPTP map $\mathcal{E}:\mathcal{D(H}_1)\rightarrow\mathcal{D(H}_2)$ is called an FPO if for any flat state $\rho\in\mathcal{D(H}_1)$ also the state $\mathcal{E}(\rho)$ is flat.

The second definition, which we will denote FPO*, is the one given in \cite{jasser2026journeyflatlanddoesantiflatness}, with a slight generalization allowing the codomain to be different from the domain. For that, let us recall the definition of the Rényi entropy \cite{Information:112906} $S_\alpha$ of a density matrix $\rho$:
\begin{equation}
S_\alpha(\rho)=\frac{1}{1-\alpha}\log\left(\tr(\rho^\alpha)\right),
\end{equation}
which can also be used to define the Rényi spread, the difference in Rényi entropies $\Delta_{\alpha\beta}$ \cite{PhysRevA.67.012326, jasser2026journeyflatlanddoesantiflatness}
\begin{equation}
\Delta_{\alpha\beta}(\rho)=S_\alpha(\rho)-S_\beta(\rho),
\end{equation}
for some $0\leq\alpha<\beta<\infty$. Then, a CPTP map $\mathcal{E}:\mathcal{D(H}_A\otimes \mathcal{H}_B)\rightarrow\mathcal{D(H}_C\otimes\mathcal{H}_D)$ is called an FPO* \cite{jasser2026journeyflatlanddoesantiflatness} if
\begin{equation}
\label{eq:FPO*}
\Delta_{\alpha\beta}\left(\tr_D[\mathcal{E}(\omega)]\right)\leq \Delta_{\alpha\beta}\left(\tr_B[\omega]\right)
\end{equation}
for all $\alpha<\beta$ and all $\omega\in\mathcal{D(H}_A\otimes\mathcal{H}_B)$. In particular, an FPO* will map flat$_A$ states to flat$_C$ states. In fact, later in Theorem 5 we show the converse to also be true: if a CPTP map takes all flat$_A$ states to flat$_C$ states, it must be FPO$^*$.

Although these two definitions look entirely different, they can actually be considered equivalent as we will show later in Theorem \ref{FPO equivalence}. In fact, a map is an FPO* if and only if the reduced map to Hilbert space $\mathcal{H}_C$ does not depend on the Hilbert space $\mathcal{H}_B$
\begin{equation}
\tr_D(\mathcal{E}(\omega))=\tr_D(\mathcal{E}[\tr_B(\omega)\otimes |0\rangle\langle 0|_B])=\mathcal{F}(\tr_B(\omega))
\end{equation}
and this induced map $\mathcal{F}:\mathcal{D(H}_A)\rightarrow\mathcal{D(H}_C)$ is an FPO.

However, we will start by showing that the result in Theorem \ref{th:convex} can be used to show that an FPO must either be a constant map or an OPO, where the appended state must be a flat state. This result is proven in Appendix \ref{Appendix: FPO}.

\begin{theorem}
\label{th:FPO}
A CPTP quantum channel $\mathcal{E:D(H}_1)\rightarrow\mathcal{D(H}_2)$ is an FPO if and only if either
\begin{enumerate}
\item $\mathcal{E}$ is a constant map to a flat state, i.e., there exists a flat state $\sigma\in \mathcal{D(H}_2)$ s.t. $\mathcal{E}(\rho)=\sigma\ \forall\rho\in\mathcal{D(H}_1)$ or
\item $\mathcal{E}(\rho)=V(\rho\otimes \sigma)V^\dagger$, for some flat state $\sigma$ in some ancillary Hilbert space $\mathcal{H}_a$ and an (injective but not necessarily bijective) isometry $V:\mathcal{H}_1\otimes\mathcal{H}_a\rightarrow \mathcal{H}_2$.
\end{enumerate}
\end{theorem}

Again, $\mathcal{H}_a$ can be 1-dimensional here, meaning that we do not actually need to append a state. Also, since we are choosing some isometry $V$ to some subspace of $\mathcal{H}_2$, we could as well restrict $\sigma$ to be a maximally mixed state in its respective Hilbert space. However, considering all flat states better illustrates the most natural ways to construct an FPO: appending an arbitrary flat state and/or performing a unitary operation.

Now, we are ready to consider the equivalence of the two definitions of Flatness-Preserving Operations. The following Theorem is proved in Appendix \ref{Appendix: FPO equivalence}.

\begin{theorem}
\label{FPO equivalence}
Let $\mathcal{E}:\mathcal{D(H}_A\otimes \mathcal{H}_B)\rightarrow\mathcal{D(H}_C\otimes\mathcal{H}_D)$ be a CPTP quantum channel and let $\omega\in\mathcal{D(H}_A\otimes\mathcal{H}_B)$. Then, the following statements are equivalent:
\begin{enumerate}
    \item $\mathcal{E}$ is an FPO*,
    \item If $\tr_B(\omega)$ is flat then $\tr_D(\mathcal{E}(\omega))$ is flat (i.e., if $\omega$ is flat$_A$ then $\mathcal{E}(\omega)$ is flat$_C$),
    \item There exists an FPO $\mathcal{F}:\mathcal{D(H}_A)\rightarrow\mathcal{D(H}_C)$ such that
    \begin{equation}
    \mathcal{F}(\tr_B (\omega))=\tr_D(\mathcal{E}(\omega))\quad \forall\omega\in\mathcal{D(H}_A\otimes\mathcal{H}_B).
    \end{equation}
\end{enumerate}
\end{theorem}

\begin{remark}
The condition 
\begin{equation}
\mathcal{F}(\tr_B (\omega))=\tr_D(\mathcal{E}(\omega))\quad \forall\omega\in\mathcal{D(H}_A\otimes\mathcal{H}_B)
\end{equation}
for some quantum channel (not necessarily FPO) $\mathcal{F}:\mathcal{D(H}_A)\rightarrow\mathcal{D(H}_C)$ is known in the literature as $\mathcal{E}$ being \emph{semicausal} meaning that the channel $\mathcal{E}$ does not allow any signaling from system $B$ to system $A$ (from Bob to Alice) \cite{PhysRevA.64.052309,priv_Hamma}. This condition is also known to be equivalent with $\mathcal{E}$ being \emph{semilocalizable} meaning that $\mathcal{E}$ can be implemented with one-way quantum communication from Alice to Bob \cite{PhysRevA.64.052309,Eggeling_2002,priv_Hamma}. Thus, an FPO* must always be both semicausal and semilocalizable.
\end{remark}

Theorem \ref{FPO equivalence} combined with theorem \ref{th:FPO} means that if a channel $\mathcal{E}:\mathcal{D(H}_A\otimes \mathcal{H}_B)\rightarrow\mathcal{D(H}_C\otimes\mathcal{H}_D)$ satisfies the FPO* condition in equation (\ref{eq:FPO*}), then either the left-hand side of that inequality must be zero for all $\omega\in\mathcal{D(H}_A\otimes\mathcal{H}_B)$
\begin{equation}
\Delta_{\alpha\beta}\left(\tr_D[\mathcal{E}(\omega)]\right)=0\leq \Delta_{\alpha\beta}\left(\tr_B[\omega]\right)
\end{equation}
or we will have an equality for all $\omega\in\mathcal{D(H}_A\otimes\mathcal{H}_B)$
\begin{equation}
\Delta_{\alpha\beta}\left(\tr_D[\mathcal{E}(\omega)]\right)= \Delta_{\alpha\beta}\left(\tr_B[\omega]\right).
\end{equation}

\medskip

\section{Conclusions}

We have shown that in order for a quantum channel to preserve the flatness-property of states, and not increase the antiflatness, the only allowed operations are any combination of the following in any order
\begin{itemize}
    \item Append a flat state,
    \item Perform a unitary operation (or, more generally, map with an isometry to another Hilbert space),
    \item Discard the state and replace with a flat state.
\end{itemize}

Our results show that this restriction also holds when considering flatness of a reduced state in a bipartite system, i.e., these are the only operations that are allowed on the reduced state. This means that in a resource theory for antiflatness, free operations are scarce, which might have been expected. To paraphrase similar catchphrases for entanglement and magic, ``antiflatness is everywhere''.

\medskip

\section*{Acknowledgments}

We thank A.~Hamma for extremely helpful comments on the draft version and for pointing out the connection to semicausality and semilocalizability. We thank N.~Pranzini for helpful discussions. OV's research is funded by the University of Helsinki Doctoral School. EKV acknowledges the financial support of the Research Council of Finland through the Finnish Quantum Flagship project (358878, UH) and the Research Council of Finland grant 1371600.

\bibliography{Flatness}
\bibliographystyle{apsrev4-2.bst}

\onecolumngrid

\appendix

\section{Orthogonality-Preserving operations}
\label{Appendix: OPO}
We will call a completely positive trace-preserving (CPTP) quantum channel $\mathcal{E:D(H}_1)\rightarrow\mathcal{D(H}_2)$ between finite Hilbert spaces an Orthogonality-Preserving Operation (OPO), if
\begin{equation}
\rho\sigma=0\Rightarrow \mathcal{E}(\rho)\mathcal{E}(\sigma)=0\quad \forall\rho,\sigma\in\mathcal{D(H}_1).
\end{equation}
\begin{thm}
\label{Appendix th:OPO}
Let $\mathcal{E:D(H}_1)\rightarrow\mathcal{D(H}_2)$ be an OPO. Let $\mathcal{H}_e$ be an ancillary Hilbert space and $W:\mathcal{H}_1\rightarrow\mathcal{H}_2\otimes\mathcal{H}_e$ an isometry such that
\begin{equation}
\mathcal{E}(\rho)=\tr_{\mathcal{H}_e}(W\rho W^\dagger).
\end{equation}
Then, there exists a state $\sigma\in \mathcal{D(H}_e)$ such that
\begin{equation}
\tr_{\mathcal{H}_2}(W\rho W^\dagger)=\sigma\quad\forall\rho\in\mathcal{D(H}_1).
\end{equation}
\end{thm}

\begin{proof}
Let us define $\mathcal{F:D(H}_1)\rightarrow\mathcal{D(H}_e)$ so that
\begin{equation}
\mathcal{F}(\rho)\equiv\tr_{\mathcal{H}_2}(W\rho W^\dagger).
\end{equation}
Now, let us pick two orthogonal states $|0\rangle, |1\rangle\in\mathcal{H}_1$ and look at the Schmidt decompositions
\begin{equation}
\label{Appendix eq:Schmidt1}
W|0\rangle_1=\sum_i \lambda_i^0 |i_0\rangle_2 |i_0\rangle_e,
\end{equation}
\begin{equation}
\label{Appendix eq:Schmidt2}
W|1\rangle_1=\sum_j \lambda_j^1 |j_1\rangle_2 |j_1\rangle_e,
\end{equation}
for some non-negative, real $\lambda_i^0, \lambda_j^1$.
From the fact that $\mathcal{E}$ is orthogonality preserving, we know that $\lambda_i^0 \lambda_j^1\langle i_0|j_1\rangle_2=0$ for all $i,j$. Next, let us consider the states $|\pm\rangle_1=\frac{1}{\sqrt{2}}(|0\rangle_1\pm|1\rangle_1)$. Then,
\begin{equation}
W|\pm\rangle_1=\frac{1}{\sqrt{2}}\left(\sum_i \lambda_i^0 |i_0\rangle_2 |i_0\rangle_e\pm\sum_j \lambda_j^1 |j_1\rangle_2 |j_1\rangle_e\right)
\end{equation}
and
\begin{equation}
\mathcal{E}(|\pm\rangle\langle\pm|)=\frac{1}{2}\bigg(\mathcal{E}(|0\rangle\langle 0|)+\mathcal{E}(|1\rangle\langle 1|)\pm\sum_{i,j}\lambda_i^0\lambda_j^1\big[{\langle i_0|j_1\rangle}_e{|j_1\rangle\langle i_0|}_2+{\langle j_1|i_0\rangle}_e{|i_0\rangle\langle j_1|}_2\big]\bigg).
\end{equation}
Since also $\mathcal{E}(|+\rangle\langle +|)$ and $\mathcal{E}(|-\rangle\langle-|)$ must be orthogonal, and also our earlier observation that $\lambda_i^0 \lambda_j^1{\langle i_0|j_1\rangle}_2=0$ for all $i,j$ ends up clearing any cross terms, we have
\begin{eqnarray}
0&=&\tr\big(\mathcal{E}(|+\rangle\langle +|)\mathcal{E}(|-\rangle\langle-|)\big)\\
&=&\frac{1}{4}\left[\tr\left(\mathcal{E}(|0\rangle\langle 0|)^2\right)+\tr\left(\mathcal{E}(|1\rangle\langle 1|)^2\right)-2\sum_{ij}{\lambda_i^0}^2{\lambda_j^1}^2|{\langle i_0|j_1\rangle}_e|^2\right]\\
&=&\frac{1}{4}\left[\tr\left(\mathcal{F}(|0\rangle\langle 0|)^2\right)+\tr\left(\mathcal{F}(|1\rangle\langle 1|)^2\right)-2\tr\left(\mathcal{F}(|0\rangle\langle 0|)\mathcal{F}(|1\rangle\langle 1|)\right)\right]\\
&=&\frac{1}{4}\tr\left(\left[\mathcal{F}(|0\rangle\langle 0|)-\mathcal{F}(|1\rangle\langle 1|)\right]^2\right).
\end{eqnarray}
Since $\mathcal{F}(|0\rangle\langle 0|)-\mathcal{F}(|1\rangle\langle 1|)$ is Hermitian, it follows that $[\mathcal{F}(|0\rangle\langle 0|)-\mathcal{F}(|1\rangle\langle 1|)]^2$ must be a positive operator and has zero trace if and only if $\mathcal{F}(|0\rangle\langle 0|)=\mathcal{F}(|1\rangle\langle 1|)$. From there, it is straightforward to see that we must have
\begin{equation}
\mathcal{F}(|\varphi\rangle\langle\varphi|)=\mathcal{F}(\mathbb{I}/d)\quad\forall |\varphi\rangle\in\mathcal{H}_1,
\end{equation}
where $d$ is the dimension of $\mathcal{H}_1$.
\end{proof}

\begin{thm}
\label{Appendix th:OPO2}
Let $\mathcal{E:D(H}_1)\rightarrow\mathcal{D(H}_2)$ be an OPO. Then there exists an ancillary Hilbert space $\mathcal{H}_a$ with a state $\sigma\in\mathcal{D(H}_a)$ and an (injective but not necessarily bijective) isometry $V:\mathcal{H}_1\otimes\mathcal{H}_a\rightarrow \mathcal{H}_2$ such that
\begin{equation}
\label{Appendix eq:isometry}
\mathcal{E}(\rho)=V\rho\otimes\sigma V^\dagger.
\end{equation}
\end{thm}
\begin{proof}
Like before, let $\mathcal{H}_e$ be an ancillary Hilbert space and $W:\mathcal{H}_1\rightarrow\mathcal{H}_2\otimes\mathcal{H}_e$ an isometry such that
\begin{equation}
\mathcal{E}(\rho)=\tr_{\mathcal{H}_e}(W\rho W^\dagger).
\end{equation}
Also, let $\sigma\in \mathcal{D(H}_e)$ be the state described in Theorem \ref{Appendix th:OPO} for which
\begin{equation}
\tr_{\mathcal{H}_2}(W\rho W^\dagger)=\sigma\quad\forall\rho\in\mathcal{D(H}_1)
\end{equation}
and let us write
\begin{equation}
\sigma=\sum_i \lambda_i^2 |i\rangle\langle i|,\quad \lambda_i>0.
\end{equation}
Based on Theorem \ref{Appendix th:OPO}, we can find for any state $|\varphi\rangle\in \mathcal{H}_1$ some orthonormal states $|\varphi_i\rangle\in\mathcal{H}_2$ such that
\begin{equation}
W |\varphi\rangle_1 =\sum_i\lambda_i {|\varphi_i\rangle}_2|i\rangle_e
\end{equation}
and
\begin{equation}
\mathcal{E}(|\varphi\rangle\langle\varphi|)=\sum_i\lambda_i^2|\varphi_i\rangle\langle\varphi_i|_2.
\end{equation}
Then, let us define $\mathcal{H}_a$ as the subspace of $\mathcal{H}_e$ spanned by the vectors $|i\rangle$ so that we can clearly restrict $\sigma$ to that subspace. Let $|j\rangle_1$ be an orthonormal basis of $\mathcal{H}_1$. Then, we can define
\begin{equation}
V=\sum_{i,j} \frac{1}{\lambda_i}\langle i|_e W|j\rangle_1 \langle j|_1\langle i|_e.
\end{equation}
Now, we need to show that $V$ is an isometry and that it satisfies equation (\ref{Appendix eq:isometry}).  The fact that $\frac{1}{\lambda_i}\langle i|_e W|j\rangle_1$ are orthonormal for all $i,j$ follows straight from the Schmidt decomposition of $W|j\rangle_1$ and that $\mathcal{E}$ is orthogonality preserving, and, consequently, $V$ is an isometry. Now, let $|\varphi\rangle_1\in\mathcal{H}_1$. Then,
\begin{equation}
V |\varphi\rangle\langle \varphi|\otimes \sigma V^\dagger=\sum_i\langle i|_e W|\varphi\rangle_1\langle \varphi|_1 W^\dagger|i\rangle_e=\sum_i\lambda_i^2|\varphi_i\rangle\langle\varphi_i|_2=\mathcal{E}(|\varphi\rangle\langle \varphi|).
\end{equation}
Therefore, equation (\ref{Appendix eq:isometry}) must hold for all pure states $|\varphi\rangle_1\in\mathcal{H}_1$ so it must hold also for all mixed states.
\end{proof}

\section{Properties of flat states}
\label{Appendix: Flat properties}

We will say that $\rho\in\mathcal{D(H)}$ is a flat state if it is proportional to a projector. However, there are multiple equivalent ways to define it.
\begin{lemma}
Let $\rho\in\mathcal{D(H)}$. The following statements are equivalent:
\begin{enumerate}
\item $\rho$ is proportional to a projector.
\item $\rho$ has only one non-zero eigenvalue (in addition to potentially having zero as an eigenvalue).
\item $\tr\left(\rho^3\right)-\tr\left(\rho^2\right)^2=0$.
\item $\rho^2=\tr(\rho^2)\rho$.
\end{enumerate}
\end{lemma}
\begin{proof}
The equivalence of the first two is trivial. Let us show that $2\Leftrightarrow 3$. Let us assume that $\rho$ has eigenvalues $\lambda_i\geq 0$. Then
\begin{eqnarray}
\tr\left(\rho^3\right)-\tr\left(\rho^2\right)^2&=&\sum_i \lambda_i^3-\left(\sum_j \lambda_j^2\right)^2\\
&=&\sum_i\lambda_i^3\sum_k\lambda_k-\sum_j\lambda_j^2\sum_l\lambda_l^2\\
&=&\sum_{i,j}\lambda_i\lambda_j\left(\lambda_i^2-\lambda_i\lambda_j\right)\\
&=&\sum_{i<j}\lambda_i\lambda_j\left(\lambda_i-\lambda_j\right)^2.
\end{eqnarray}
Clearly, the above equation is zero if and only if every non-zero eigenvalue is the same. Next, let us show that $2\Leftrightarrow 4$. Let $\rho$ have spectral decomposition
\begin{equation}
\rho=\sum_i\lambda_i|\varphi_i\rangle\langle\varphi_i|.
\end{equation}
Then,
\begin{eqnarray}
\rho^2-\tr(\rho^2)\rho&=&\sum_i\lambda_i^2|\varphi_i\rangle\langle\varphi_i|-\sum_j\lambda_j^2\sum_k\lambda_k|\varphi_k\rangle\langle\varphi_k|\\
&=&\sum_{i,j}\left(\lambda_i^2\lambda_j-\lambda_j^2\lambda_i\right)|\varphi_i\rangle\langle\varphi_i|\\
&=&\sum_{i,j}\lambda_i\lambda_j(\lambda_i-\lambda_j)|\varphi_i\rangle\langle\varphi_i|.
\end{eqnarray}
Clearly, this is zero if and only if every non-zero eigenvalue is the same.
\end{proof}

\section{Convex combinations of flat states}
\label{Appendix: Convex combinations}

\begin{thm}
\label{Appendix th:convex}
Let $\rho_1, \rho_2\in \mathcal{D(H)}$ be two flat states and $\rho_1\neq\rho_2$ and let $0<p<1$. Then, the state
\begin{equation}
\sigma=p\rho_1+(1-p)\rho_2
\end{equation}
is flat if and only if
\begin{equation}
\rho_1\rho_2=0,
\end{equation}
(i.e., $\rho_1$ and $\rho_2$ have orthogonal support) and
\begin{equation}
p\tr(\rho_1^2)-(1-p)\tr(\rho_2^2)=0.
\end{equation}
\end{thm}
\begin{proof}
Let us start with the slightly easier $\Leftarrow$ case so we will assume
\begin{equation}
\rho_1\rho_2=0,\qquad\qquad p\tr(\rho_1^2)-(1-p)\tr(\rho_2^2)=0.
\end{equation}
Then,
\begin{equation}
\sigma=p\rho_1+(1-p)\rho_2=\frac{\tr(\rho_2^2)\rho_1+\tr(\rho_1^2)\rho_2}{\tr(\rho_1^2)+\tr(\rho_2^2)}.
\end{equation}
Since $\rho_1$ and $\rho_2$ have orthogonal support, the only non-zero eigenvalue of $\sigma$ will be
\begin{equation}
\tr(\sigma^2)=\frac{\tr(\rho_1^2)\tr(\rho_2^2)}{\tr(\rho_1^2)+\tr(\rho_2^2)}
\end{equation}
and $\sigma$ will be a flat state.

Let us then focus on the $\Rightarrow$. Let us first show that for flat $\rho_1, \rho_2$,
\begin{equation}
0\leq\tr(\rho_1\rho_2)\leq\min(\tr(\rho_1^2),\tr(\rho_2^2)).
\end{equation}
The first inequality is easily shown true for any positive operators. For the second, let $|i\rangle$ be the eigenvectors of $\rho_1$ corresponding to the non-zero eigenvalue. Then
\begin{equation}
\tr(\rho_1\rho_2)=\tr(\rho_1^2)\sum_i\langle i|\rho_2|i\rangle\leq\tr(\rho_1^2)\tr(\rho_2)=\tr(\rho_1^2)
\end{equation}
and similarly for $\rho_2$, so clearly we also have $\tr(\rho_1\rho_2)\leq\min(\tr(\rho_1^2),\tr(\rho_2^2))$.

Now, a state $\sigma$ is flat if and only if $\mathcal{F}(\sigma)=\tr(\sigma^3)-\tr^2(\sigma^2)=0$. Let us write $r\equiv\tr(\rho_1^2)$ and $s\equiv\tr(\rho_2^2)$ and compute $\mathcal{F}(\sigma)$ for $\sigma=p\rho_1+(1-p)\rho_2$:
\begin{eqnarray}
\mathcal{F}(\sigma)&=&p(1-p)\left(pr-[1-p]s\right)^2\nonumber\\
&&+p(1-p)\left(3pr+3[1-p]s-4p^2r-4[1-p]^2s\right)\tr(\rho_1\rho_2)\nonumber\\
&&-4p^2(1-p)^2\tr^2(\rho_1\rho_2).
\end{eqnarray}
Setting this to zero and solving for $\tr(\rho_1\rho_2)$ gives us
\begin{eqnarray}
\label{Appendix eq:trace12}
\tr(\rho_1\rho_2)&=&\frac{(3p-4p^2)r+\left(3[1-p]-4[1-p]^2\right)s}{8p(1-p)}\nonumber\\
&&\pm \frac{\sqrt{(-8p^3+9p^2)r^2+\left(-8[1-p]^3+9[1-p]^2\right)s^2-6p(1-p)rs}}{8p(1-p)}
\end{eqnarray}
Now, we know that $0\leq \tr(\rho_1\rho_2)\leq \min(r,s)$. We can write the radicand above as
\begin{eqnarray}
&&(-8p^3+9p^2)r^2+\left(-8[1-p]^3+9[1-p]^2\right)s^2-6p(1-p)rs\nonumber\\
&=&\left[(3p-4p^2)r+\left(3[1-p]-4[1-p]^2\right)s\right]^2+16p(1-p)\left[pr-(1-p)s\right]^2
\end{eqnarray}
so clearly the solution of (\ref{Appendix eq:trace12}) with the minus sign can only satisfy $\tr(\rho_1\rho_2)\geq 0$ if and only if $pr=(1-p)s$ which gives $\tr(\rho_1\rho_2)=0\Rightarrow \rho_1\rho_2=0$.

Without loss of generality, let us assume that $r\leq s$. For the solution of (\ref{Appendix eq:trace12}) with the plus sign, we will use the upper bound:
\begin{eqnarray}
0\leq r-\tr(\rho_1\rho_2)&=&\frac{(5p-4p^2)r-\left(3[1-p]-4[1-p]^2\right)s}{8p(1-p)}\nonumber\\
&&-\frac{\sqrt{(-8p^3+9p^2)r^2+\left(-8[1-p]^3+9[1-p]^2\right)s^2-6p(1-p)rs}}{8p(1-p)}.
\end{eqnarray}
Now, we can write the radicand as
\begin{eqnarray}
&&(-8p^3+9p^2)r^2+\left(-8[1-p]^3+9[1-p]^2\right)s^2-6p(1-p)rs\nonumber\\
&=&\left[(5p-4p^2)r-\left(3[1-p]-4[1-p]^2\right)s\right]^2+16p(1-p)^2\left[pr+(1-p)s\right](s-r).
\end{eqnarray}
The bound can now be satisfied only if $s=r$, which would give $\tr(\rho_1\rho_2)=r$ implying $\rho_1=\rho_2$, contradicting our assumption. Therefore, the only solution is the one with $pr=(1-p)s$ and $\rho_1\rho_2=0$.
\end{proof}

\section{Flatness-Preserving Operations are either constant or unitary}
\label{Appendix: FPO}
Here, a Flatness-Preserving Operation (FPO) means a completely positive trace-preserving (CPTP) map $\mathcal{E}:\mathcal{D(H}_1)\rightarrow\mathcal{D(H}_2)$ such that for any flat $\rho\in\mathcal{D(H}_1)$ also $\mathcal{E}(\rho)$ is flat.
\begin{lemma}
\label{Appendix th:orthogonals}
Let quantum channel $\mathcal{E:D(H}_1)\rightarrow\mathcal{D(H}_2)$ be an FPO and $\rho_1, \rho_2\in\mathcal{D(H}_1)$ some flat states orthogonal to each other. Then either $\mathcal{E}(\rho_1)=\mathcal{E}(\rho_2)$ or we have both 
\begin{align}
\mathcal{E}(\rho_1)\mathcal{E}(\rho_2)&=0 &&\text{and} &\frac{\tr(\mathcal{E}(\rho_1)^2)}{\tr(\mathcal{E}(\rho_2)^2)}&=\frac{\tr(\rho_1^2)}{\tr(\rho_2^2)}.
\end{align}
\end{lemma}
\begin{proof}
Let $r=\tr(\rho_1^2)$ and $s=\tr(\rho_2^2)$ and let $p=\frac{s}{r+s}$. Then, Theorem \ref{Appendix th:convex} tells us that $\sigma=p\rho_1+(1-p)\rho_2$ is a flat state meaning that $\mathcal{E}(\sigma)$ must be a flat state as well. But since $\mathcal{E}(\sigma)=p\mathcal{E}(\rho_1)+(1-p)\mathcal{E}(\rho_2)$ we know from Theorem \ref{Appendix th:convex} that either $\mathcal{E}(\rho_1)=\mathcal{E}(\rho_2)$ or we have that $\mathcal{E}(\rho_1)\mathcal{E}(\rho_2)=0$ and $p\tr(\mathcal{E}(\rho_1)^2)-(1-p)\tr(\mathcal{E}(\rho_2)^2)=0$ from which the claim follows.
\end{proof}
\begin{lemma}
\label{Appendix th:twodim}
Let $\mathcal{H}_1$ be a 2-dimensional Hilbert space with basis states $|0\rangle$ and $|1\rangle$ and let $\mathcal{E:D(H}_1)\rightarrow\mathcal{D(H}_2)$ an FPO with 
\begin{equation}
\mathcal{E}(|0\rangle\langle 0|)=\mathcal{E}(|1\rangle\langle 1|)=\sigma\in\mathcal{D(H}_2).
\end{equation}
Then,
\begin{equation}
\mathcal{E}(\rho)=\sigma\quad\forall \rho\in\mathcal{D(H}_1).
\end{equation}
\end{lemma}
\begin{proof}
Let $X, Y, Z$ be the Pauli matrices. Clearly $\mathcal{E}(\mathbb{I}/2)=\sigma$ and $\mathcal{E}(Z)=0$. Let us assume that $\mathcal{E}(X)\neq 0$. Then,
\begin{equation}
\mathcal{E}\left(\frac{\mathbb{I}+\cos\theta X+\sin\theta Z}{2}\right)=\frac{1+\cos\theta}{2}\mathcal{E}\left(\frac{\mathbb{I}+X}{2}\right)+\frac{1-\cos\theta}{2}\mathcal{E}\left(\frac{\mathbb{I}-X}{2}\right)
\end{equation}
should be a flat state for any value of $\theta$ but this contradicts Theorem \ref{Appendix th:convex}. Therefore, $\mathcal{E}(X)=0$ and similarly $\mathcal{E}(Y)=0$ and $\mathcal{E}(\rho)=\sigma$ for all $\rho\in\mathcal{D(H}_1)$.
\end{proof}
\begin{lemma}
Let $\mathcal{H}_1$ be a n-dimensional Hilbert space with a basis $|i\rangle$, $i\in\{0,\dots,n-1\}$ and let $\mathcal{E:D(H}_1)\rightarrow\mathcal{D(H}_2)$ an FPO with 
\begin{equation}
\mathcal{E}(|0\rangle\langle 0|)=\mathcal{E}(|1\rangle\langle 1|)=\sigma\in\mathcal{D(H}_2).
\end{equation}
Then,
\begin{equation}
\mathcal{E}(|i\rangle\langle i|)=\sigma\quad\forall i\in\{0,\dots,n-1\}.
\end{equation}
\end{lemma}
\begin{proof}
Let us assume that $\mathcal{E}(|j\rangle\langle j|)\neq \sigma$. Then, Lemma \ref{Appendix th:orthogonals} tells us that
\begin{equation}
\frac{\tr([|0\rangle\langle 0|]^2)}{\tr([|j\rangle\langle j|]^2)}=\frac{\tr(\sigma^2)}{\tr(\mathcal{E}(|j\rangle\langle j|)^2)}=\frac{\tr\left(\left[\frac{|0\rangle\langle 0|+|1\rangle\langle 1|}{2}\right]^2\right)}{\tr([|j\rangle\langle j|]^2)}.
\end{equation}
However, the left-hand side is 1 while the right hand side is 1/2, so this is clearly a contradiction and we must have $\mathcal{E}(|i\rangle\langle i|)=\sigma$ for all $i\in\{0,\dots,n-1\}$.
\end{proof}
\begin{lemma}
Let $\mathcal{H}_1$ be a Hilbert space with two orthogonal subspaces $S_1$ and $S_2$. Let $\mathcal{E:D(H}_1)\rightarrow\mathcal{D(H}_2)$ be an FPO with
\begin{equation}
\mathcal{E}(\rho_1)=\mathcal{E}(\rho_2)=\sigma\quad\forall\rho_1\in\mathcal{D}(S_1)\ \forall\rho_2\in\mathcal{D}(S_2).
\end{equation}
Then
\begin{equation}
\mathcal{E}(\rho)=\sigma\quad\forall\rho\in\mathcal{D}(S_1\oplus S_2).
\end{equation}
\end{lemma}
\begin{proof}
Let us pick some state $|\varphi\rangle\in S_1\oplus S_2$, $|\varphi\rangle\not\in S_1\cup S_2$. The projection of $|\varphi\rangle$ onto $S_1$ and the projection of $|\varphi\rangle$ onto $S_2$ form an orthogonal basis for a two-dimensional subspace and Lemma \ref{Appendix th:twodim} tells us that $\mathcal{E}(|\varphi\rangle\langle\varphi|)=\sigma$. Therefore, the claim holds for all pure states, so it must also hold for all mixed states.
\end{proof}
\begin{corollary}
Let $\mathcal{E:D(H}_1)\rightarrow\mathcal{D(H}_2)$ be an FPO and $\rho_1, \rho_2\in\mathcal{D(H}_1)$ flat states orthogonal to each other and $\mathcal{E}(\rho_1)=\mathcal{E}(\rho_2)=\sigma$. Then
\begin{equation}
\mathcal{E}(\rho)=\sigma\quad\forall\rho\in\mathcal{D(H}_1).
\end{equation}
\end{corollary}
\begin{corollary}
Let $\mathcal{E:D(H}_1)\rightarrow\mathcal{D(H}_2)$ be an FPO and $\rho_1, \rho_2\in\mathcal{D(H}_1)$ flat states orthogonal to each other and $\mathcal{E}(\rho_1)\neq\mathcal{E}(\rho_2)$. Then $\mathcal{E}$ preserves orthogonality of all pairs of states and preserves the proportions of purities for all pairs of flat states.
\end{corollary}
\begin{proof}
Orthogonal pure states are mapped to orthogonal flat states, so it follows immediately that arbitrary orthogonal states will stay orthogonal as well. Also, since the proportions of purity are preserved for orthogonal flat states, we have that
\begin{equation}
\frac{\tr\left(\left[\mathcal{E}\left(\frac{|0\rangle\langle 0|+|1\rangle\langle 1|}{2}\right)\right]^2\right)}{\tr(\mathcal{E}[|0\rangle\langle 0|]^2)}=\frac{\tr\left(\left[\mathcal{E}\left(\frac{|0\rangle\langle 0|}{2}\right)\right]^2+\left[\mathcal{E}\left(\frac{|1\rangle\langle 1|}{2}\right)\right]^2\right)}{\tr(\mathcal{E}[|0\rangle\langle 0|]^2)}=\frac{\tr\left(\left[\frac{|0\rangle\langle 0|}{2}\right]^2+\left[\frac{|1\rangle\langle 1|}{2}\right]^2\right)}{\tr([|0\rangle\langle 0|]^2)}=\frac{\tr\left(\left[\frac{|0\rangle\langle 0|+|1\rangle\langle 1|}{2}\right]^2\right)}{\tr([|0\rangle\langle 0|]^2)},
\end{equation}
and similarly it easily follows that the proportion of purity is preserved also for any flat state with the maximally mixed state, and thus it is preserved between any pair of flat states.
\end{proof}

\begin{thm}
\label{Appendix th:FPO}
A CPTP quantum channel $\mathcal{E:D(H}_1)\rightarrow\mathcal{D(H}_2)$ is an FPO if and only if either
\begin{enumerate}
\item $\mathcal{E}$ is a constant map to a flat state, i.e., there exists a flat state $\sigma\in \mathcal{D(H}_2)$ s.t. $\mathcal{E}(\rho)=\sigma\ \forall\rho\in\mathcal{D(H}_1)$ or
\item $\mathcal{E}(\rho)=V(\rho\otimes \sigma)V^\dagger$, for some flat state $\sigma$ in some ancillary Hilbert space $\mathcal{H}_a$ and an (injective but not necessarily bijective) isometry $V:\mathcal{H}_1\otimes\mathcal{H}_a\rightarrow \mathcal{H}_2$.
\end{enumerate}
\end{thm}
\begin{proof}
We have already shown that an FPO is either a constant map or an OPO. We also know from Theorem \ref{Appendix th:OPO2} that an OPO must be of the form $\mathcal{E}(\rho)=V(\rho\otimes \sigma)V^\dagger$. The only thing left to show is that, for an FPO, $\sigma$ must be a flat state, but, since the eigenvalues of $\rho\otimes\sigma$ are products of eigenvalues of $\rho$ and $\sigma$, this follows trivially.
\end{proof}

\section{Equivalence of the Flatness-Preserving Operation definitions}
\label{Appendix: FPO equivalence}
Again, let Flatness-Preserving Operation (FPO) mean a completely positive trace-preserving (CPTP) map $\mathcal{E}:\mathcal{D(H}_1)\rightarrow\mathcal{D(H}_2)$ such that for any flat $\rho\in\mathcal{D(H}_1)$ also $\mathcal{E}(\rho)$ is flat.

As an alternative definition, we will call a CPTP map $\mathcal{E}:\mathcal{D(H}_A\otimes \mathcal{H}_B)\rightarrow\mathcal{D(H}_C\otimes\mathcal{H}_D)$ an FPO* if
\begin{equation}
\label{Appendix eq:FPO*}
\Delta_{\alpha\beta}\left(\tr_D[\mathcal{E}(\omega)]\right)\leq \Delta_{\alpha\beta}\left(\tr_B[\omega]\right)\qquad \forall\alpha<\beta\quad\forall\omega\in\mathcal{D(H}_A\otimes\mathcal{H}_B),
\end{equation}
where
\begin{equation}
\Delta_{\alpha\beta}(\rho)=S_\alpha(\rho)-S_\beta(\rho)\qquad 0\leq\alpha<\beta<\infty
\end{equation}
and $S_\alpha, S_\beta$ are Rényi entropies
\begin{equation}
S_\alpha(\rho)=\frac{1}{1-\alpha}\log\left(\tr(\rho^\alpha)\right).
\end{equation}
Note that, for a flat state $\rho$ we have
\begin{equation}
    S_\alpha(\rho)=-\log\left(\tr(\rho^2)\right)\qquad\Delta_{\alpha\beta}(\rho)=0\qquad\forall\alpha<\beta
\end{equation}
and for any state $\rho$
\begin{equation}
\Delta_{\alpha\beta}(\rho)\geq 0\qquad\forall\alpha<\beta,
\end{equation}
where the equality holds only for flat states.

Also, for any states $\rho, \sigma$ we have
\begin{equation}
S_\alpha(\rho\otimes\sigma)=S_\alpha(\rho)+S_\alpha(\sigma)\qquad \Delta_{\alpha\beta}(\rho\otimes\sigma)=\Delta_{\alpha\beta}(\rho)+\Delta_{\alpha\beta}(\sigma)\qquad\forall\alpha<\beta.
\end{equation}

\begin{thm}
Let $\mathcal{E}:\mathcal{D(H}_A\otimes \mathcal{H}_B)\rightarrow\mathcal{D(H}_C\otimes\mathcal{H}_D)$ be a CPTP quantum channel and let $\omega\in\mathcal{D(H}_A\otimes\mathcal{H}_B)$. Then, the following statements are equivalent:
\begin{enumerate}
    \item $\mathcal{E}$ is an FPO*,
    \item If $\tr_B(\omega)$ is flat then $\tr_D(\mathcal{E}(\omega))$ is flat (i.e., if $\omega$ is flat$_A$ then $\mathcal{E}(\omega)$ is flat$_C$),
    \item There exists an FPO $\mathcal{F}:\mathcal{D(H}_A)\rightarrow\mathcal{D(H}_C)$ such that
    \begin{equation}
    \mathcal{F}(\tr_B (\omega))=\tr_D(\mathcal{E}(\omega))\quad \forall\omega\in\mathcal{D(H}_A\otimes\mathcal{H}_B).
    \end{equation}
\end{enumerate}
\end{thm}
\begin{proof}
The implication $1\Rightarrow 2$ is trivial. Let us show $3\Rightarrow 1$. Let $\mathcal{E}:\mathcal{D(H}_A\otimes \mathcal{H}_B)\rightarrow\mathcal{D(H}_C\otimes\mathcal{H}_D)$ be a CPTP quantum channel and let there be an FPO $\mathcal{F}:\mathcal{D(H}_A)\rightarrow\mathcal{D(H}_C)$ such that
\begin{equation}
\mathcal{F}(\tr_B (\omega))=\tr_D(\mathcal{E}(\omega))\quad \forall\omega\in\mathcal{D(H}_A\otimes\mathcal{H}_B).
\end{equation}
Since $\mathcal{F}$ is an FPO, we know from Theorem \ref{Appendix th:FPO} that either $\mathcal{F}(\tr_B(\omega))$ is a fixed flat state for all $\omega$ or it has the same spectrum as $\tr_B(\omega)\otimes\sigma$ for some fixed flat state $\sigma$. In the first case,
\begin{equation}
\Delta_{\alpha\beta}(\tr_D[\mathcal{E}(\omega)])=\Delta_{\alpha\beta}(\mathcal{F}[\tr_B(\omega)])=0\qquad \forall\omega\in\mathcal{D(H}_A\otimes\mathcal{H}_B),\ \forall\alpha,\beta,
\end{equation}
which would make $\mathcal{E}$ an FPO*. In the second case
\begin{equation}
\Delta_{\alpha\beta}(\tr_D[\mathcal{E}(\omega)])=\Delta_{\alpha\beta}(\mathcal{F}[\tr_B(\omega)])=\Delta_{\alpha\beta}(\tr_B(\omega)\otimes\sigma)=\Delta_{\alpha\beta}(\tr_B(\omega))\qquad \forall\omega\in\mathcal{D(H}_A\otimes\mathcal{H}_B),\ \forall\alpha,\beta,
\end{equation}
which, again, would make $\mathcal{E}$ an FPO*.

Finally, let us prove $2\Rightarrow 3$. Let $\mathcal{E}:\mathcal{D(H}_A\otimes \mathcal{H}_B)\rightarrow\mathcal{D(H}_C\otimes\mathcal{H}_D)$ be a CPTP quantum channel for which any flat$_A$ state is mapped to a flat$_C$ state. Let us define new channels $\mathcal{F}^A_{\sigma}:\mathcal{D(H}_A)\rightarrow\mathcal{D(H}_C)$, $\mathcal{F}^B_{\rho}:\mathcal{D(H}_B)\rightarrow\mathcal{D(H}_C)$ for each $\sigma\in\mathcal{D(H}_B), \rho\in\mathcal{D(H}_A)$ so that
\begin{equation}
\mathcal{F}^A_\sigma(\rho)=\mathcal{F}^B_\rho(\sigma)=\tr_D\left[\mathcal{E}(\rho\otimes\sigma)\right]
\end{equation}
Clearly, these must be CPTP, since the original $\mathcal{E}$ is CPTP. Furthermore, if $\rho\in\mathcal{D(H}_A)$ is a flat state, then also $\mathcal{F}^A_\sigma(\rho)=\mathcal{F}^B_\rho(\sigma)$ must be a flat state for any $\sigma$ (since we assume that $\mathcal{E}$ maps flat$_A$ states to flat$_C$ states). Therefore, $\mathcal{F}^A_\sigma$ must be an FPO for every index $\sigma\in\mathcal{D(H}_B)$ and $\mathcal{F}^B_\rho$ must be an FPO for every index \emph{flat} $\rho\in\mathcal{D(H}_A)$. However, since for a flat $\rho$ the map $\mathcal{F}^B_\rho(\sigma)$ must be flat also for non-flat $\sigma$, according to Theorem \ref{Appendix th:FPO} it must be a constant map to a flat state. Therefore, we also know that for a flat $\rho$ and any $\sigma_1,\sigma_2\in\mathcal{D(H}_B)$
\begin{equation}
\mathcal{F}^A_{\sigma_1}(\rho)=\mathcal{F}^A_{\sigma_2}(\rho).
\end{equation}
In particular, this will hold for any pure state $\rho$ so it will also hold more generally for any state. Therefore, we can ignore the subindex $\sigma$ in the map $\mathcal{F}^A_{\sigma}$ to define $\mathcal{F}^A$ using $|0\rangle\langle 0|$:
\begin{equation}
\tr_D[\mathcal{E}(\rho\otimes\sigma)]=\mathcal{F}^A_{\sigma}(\rho)=\mathcal{F}^A(\rho)=\tr_D[\mathcal{E}(\rho\otimes|0\rangle\langle 0|)].
\end{equation}

Since separable states span the whole set of states (or the whole space of linear operators for that matter), the map $\tr_D[\mathcal{E}(\cdot)]:\mathcal{D(H}_A\otimes\mathcal{H}_B)\rightarrow\mathcal{D(H}_C)$ is also fully defined by its action on the separable states. Therefore, we can conclude that for any $\omega\in\mathcal{D(H}_A\otimes\mathcal{H}_B)$, we can write $\omega=\sum_i c_i \rho_i\otimes\sigma_i$ for some $c_i\in \mathbb{C}$, $\rho_i\in \mathcal{D(H}_A)$, $\sigma_i\in\mathcal{D(H}_B)$ and
\begin{equation}
\tr_D[\mathcal{E}(\omega)]=\tr_D\left[\mathcal{E}\left(\sum_i c_i \rho_i\otimes\sigma_i\right)\right]=\tr_D\left[\mathcal{E}\left(\sum_i c_i \rho_i\otimes|0\rangle\langle 0|\right)\right]=\tr_D[\mathcal{E}(\tr_B[\omega]\otimes|0\rangle\langle 0|)]=\mathcal{F}^A(\tr_B(\omega)),
\end{equation}
and we have already shown that $\mathcal{F}^A$ is an FPO.
\end{proof}

\end{document}